\documentclass[12pt,onecolumn,draftcls]{IEEEtran}
\usepackage{amsmath,epsfig}
\usepackage{amssymb}
\usepackage{amsmath}
\usepackage{cases} 
\usepackage{amssymb,amsmath,cite}
\usepackage{epsfig}
\usepackage{color}
\usepackage{slashbox}
\usepackage{bm}

\renewcommand{\frac}{\dfrac}

\newcommand{\K}{{\cal K}}
\newcommand{\N}{{\cal N}}

\newcommand{\A}{{\cal A}}

\renewcommand{\S}{\cal S}
\newcommand{\SI}{\mbox{SNR}}

\newcommand{\erhao}{\fontsize{23pt}{\baselineskip}\selectfont}

\newtheorem{dingli}{Theorem~}

\newenvironment{proof}[1][Proof]{\begin{trivlist}
\item[\hskip \labelsep {\bfseries #1}]}{\end{trivlist}}

\begin{document}
  \title{\erhao{Complexity Analysis of Joint Subcarrier and Power Allocation for the Cellular Downlink OFDMA System}
  \author{Ya-Feng Liu}
  \thanks{This work was supported by the National Natural Science Foundation, Grants 11331012 and 11301516.}
 \thanks{Y.-F.~Liu is with the State Key Laboratory
of Scientific and Engineering Computing, Institute of Computational
Mathematics and Scientific/Engineering Computing, Academy of
Mathematics and Systems Science, Chinese Academy of Sciences,
Beijing, 100190, China (e-mail:
{{yafliu}@lsec.cc.ac.cn}).}
}
  \maketitle

   \begin{abstract}
     \boldmath Consider the cellular downlink Orthogonal Frequency Division Multiple Access (OFDMA) system where a single transmitter transmits signals to multiple receivers on multiple discrete subcarriers. To adapt fast channel fluctuations, the transmitter should be able to dynamically allocate subcarrier and power resources. Assuming perfect channel knowledge, we formulate the joint subcarrier and power allocation problem as two optimization problems: the first is the one of minimizing the total transmission power subject to quality of service constraints, and the second is the one of maximizing a system utility function subject to power budget constraints. In this letter, we show that both the aforementioned formulations of the joint subcarrier and power allocation problem are generally NP-hard. We also identify several subclasses of the problem which are polynomial time solvable.





 %

     \end{abstract}
\begin{keywords}
Computational complexity, cellular downlink OFDMA system, power
control, subcarrier allocation.
\end{keywords}

   \section{Introduction}
Orthogonal Frequency Division Multiple Access (OFDMA) is a form of multi-carrier transmission and is well suited for frequency selective
channels and high data rates. To adapt channel fluctuations in space and time and improve the overall system's throughput, OFDMA based systems should be equipped with {dynamic} subcarrier and power allocation algorithms. Recently, various heuristics approaches have been proposed for the joint subcarrier and power allocation problem for the OFDMA system {\cite{dynamic,wong,broadcast,JSAC,outage,personal,newe,sub-optimal,newh}}. However, none of them claimed that they could solve the problem (except some special cases) to global optimality in polynomial time and the reasons are often attributed to nonconvexity of the problem. 
However, not all nonconvex problems are hard to solve since the lack of convexity may be due to an inappropriate
formulation, and many nonconvex optimization problems indeed admit a convex reformulation; some examples can be found in \cite{complexity,precoding,simo,coordinated,Coordition,RQ}. In contrast to nonconvexity, computational complexity theory\cite{Complexitybook} can characterize inherent tractability of an optimization problem. The goal of this letter is to analyze the computational complexity of the joint subcarrier and power allocation problem for the cellular downlink OFDMA system.

The computational complexity of the dynamic spectrum management problem for the interference channel (IC) has been extensively studied in \cite{complexity}. It is shown there that the dynamic spectrum management problem is NP-hard when the number of subcarriers is greater than two, or when the number of users is greater than one. For the OFDMA system,  the reference \cite{hayashi} showed that the sum-rate maximization problem is NP-hard. 
Very recently, \cite{multi-user} provided a systematic characterization of the computational complexity status of the joint subcarrier and power allocation problem for the {multi-user OFDMA system}, {where multiple transmitters transmit signals to multiple receivers on multiple subcarriers and different transmitters are not allowed to share transmission power}. 

%
%
In this letter, we focus on the characterization of the computational complexity of the joint subcarrier and power allocation problem for the {cellular downlink OFDMA system}. We consider two formulations of the problem: the total power minimization formulation and the system utility maximization formulation. The contributions of this letter are twofold. First, we show that both formulations of the problem are generally NP-hard. The NP-hardness results suggest that for a given cellular downlink OFDMA system, finding the optimal subcarrier and power allocation strategy is computationally intractable. 
Second, we identify several subclasses of the problem which can be solved in polynomial time. 
In particular, we show that the sum-rate maximization problem for the cellular downlink OFDMA system is polynomial time solvable. This result is in sharp contrast to the ones in \cite{complexity,hayashi,multi-user,add2}, {where the sum-rate maximization problem is shown to be NP-hard in various different scenarios.}
%

\section{System Model and Problem Formulation}

{In this section, we introduce the system model and problem formulation.} Consider a cellular downlink OFDMA system, where a single transmitter (base station) transmits signals to $K$ receivers on $N$ subcarriers. Throughout this letter, we assume that $N\geq K$; i.e., the number of subcarriers is greater than or equal to the number of receivers. 

Denote the set of receivers and the set of subcarriers by $\K=\left\{1,2,...,K\right\}$ and $\N=\left\{1,2,...,N\right\}$, respectively. For any $k\in\K$ and $n\in\N$, suppose $s_k^n$ to be the {complex} symbol that the transmitter wishes to send to receiver $k$ on subcarrier $n$, then the received signal $\hat s_k^n$ at receiver $k$ on subcarrier $n$ can be expressed by
${\hat s_k^n=h_{k}^n s_k^n + z_k^n,}$ where $h_{k}^n$ is the {complex} channel coefficient from the transmitter to receiver $k$ on subcarrier $n$ and $z_k^n$ is the {complex} additive white Gaussian
noise with distribution $\cal{CN}$$({0}, \eta_k^n).$
Denoting the power of $s_k^n$ by $p_k^n$; i.e., $p_k^n:=|s_k^n|^2$, the received power at receiver $k$ on subcarrier $n$ is given by
${\alpha_{k}^np_k^n+\eta_k^n},k\in\K,n\in\N,$ where $\alpha_{k}^n:=|h_{k}^n|^2$ 
stands for the channel gain from the transmitter to receiver $k$ on subcarrier $n.$ Then, we can write the achievable data rate $R_k$ of receiver $k$ as 
\begin{equation*}\label{rk}
{R_k=\sum_{n\in\N}\log_2\left(1+\SI_k^n\right)=\sum_{n\in\N}\log_2\left(1+\frac{\alpha_{k}^np_k^{n}}{\eta_k^n}\right)},~k\in\K.\end{equation*}
%


In this letter, {we assume that the transmitter knows all channel gains (through either feeback or reverse link estimation)}.
We consider the joint subcarrier and power allocation problem for the cellular downlink OFDMA system:\vspace{-0.2cm}
  \begin{subequations}\label{problem}
 \begin{eqnarray}
 &\displaystyle\min_{\{p_k^n\}} & \displaystyle \sum_{k\in\K}\sum_{n\in\N}p_k^{n} \label{objective}\\
&\text{s.t.} & R_{k}\geq \gamma_{k},~k\in\K, \label{qos}\\
        && P^n\geq p_k^n\geq 0,~k\in\K,~n\in\N,\label{subcarriercons}\\
    && p_k^np_j^n=0,~\forall~j\neq k,{~k,\,j\in\K,}~n\in\N,\label{ofdma}
    \end{eqnarray}
\end{subequations}
where the objective function in \eqref{objective} is the total transmission power, $\gamma_k>0$ in \eqref{qos} is the desired transmission rate target of receiver $k,$ $P^n$ in \eqref{subcarriercons} is the transmission power budget on subcarrier $n$, and the last OFDMA constraint \eqref{ofdma} requires that the transmitter is allowed to transmit signals to at most one receiver on each subcarrier.

Besides the above total transmission power minimization problem, we also consider the system utility maximization problem, which can be expressed by
\begin{equation}\label{utility}
 \begin{array}{cl}
\displaystyle \max_{\{p_k^n\}} & \displaystyle H(R_1,R_2,...,R_K) \\
\text{s.t.}  & \displaystyle \eqref{subcarriercons}, \eqref{ofdma},~\text{and}~\sum_{k\in\K}\sum_{n\in\N}p_k^{n}\leq P, \\
    \end{array}
\end{equation}where $P$ is the power budget of the transmitter, and $H(R_1,R_2,...,R_K)$ denotes the system utility function. Four popular system utility functions are:
\begin{itemize}
\item [-] {Sum-rate utility:}
$H_1={{\sum_{k=1}^K{R_k}}}/{K},$
\item [-] {Proportional fairness utility:} $H_2=\left(\prod_{k=1}^KR_k\right)^{1/K},$\\[3pt]
\item [-] {Harmonic mean utility:}
$H_3=K/\left(\sum_{k=1}^K R_k^{-1}\right),$\\[-3pt]
\item [-] {Min-rate utility:} $H_4=\displaystyle \min_{1\leq k\leq
K}\left\{R_k\right\}.$
\end{itemize}


\section{Complexity Analysis}

  In this section, we shall investigate the computational complexity of problems \eqref{problem} and \eqref{utility}. We shall first show in Subsection \ref{sub1} that both problem \eqref{problem} and problem \eqref{utility} with $H=H_2,H_3,H_4$ are NP-hard when the ratio of the number of subcarriers and the number of users, that is $N/K$, is equal to any constant $c>1$. Then, we shall identify some easy subclasses of problems \eqref{problem} and \eqref{utility} which are polynomial time solvable in Subsection \ref{sub2}. {In particular, we shall show that problem \eqref{utility} with $H=H_1$ is polynomial time solvable.} 



\subsection{Hard Cases}\label{sub1}
 We first show the NP-hardness of problem \eqref{problem}. To do this, we consider its feasibility problem. If the feasibility problem
is NP-hard, so is the original optimization problem. The NP-hardness proof of the feasibility problem of \eqref{problem} is based on a polynomial time transformation from the 3-partition problem.

%

\begin{dingli}\label{thm-power}
  Given any constant $c>1,$ checking the feasibility of problem \eqref{problem} is NP-hard when $N/K=c$. Thus, problem \eqref{problem} is also NP-hard when $N/K=c$.
\end{dingli}
   \begin{proof}{Without loss of generality, consider the special case of problem \eqref{problem} where $c=3$\footnote{{In fact, any case where $c>1$ can be reduced to the case where $c=3$ by introducing some dummy receivers and/or subcarriers. This can be done in a similar fashion as in \cite{multi-user}.}}}. Given any instance of the $3$-partition problem with $\A=\left\{a_1,a_2,...,a_{3K}\right\}$ and $B\in\mathbb{Z}^+,$ where $B/4<a_n<B/2$ for each $a_n\in\A$ and $\sum_{n=1}^{3K}a_n=KB,$ 
%
   we construct a cellular downlink OFDMA system where there are $3K$ subcarriers and $K$ receivers. Hence, $\K=\left\{1,2,...,K\right\}$ and $\N=\left\{1,2,...,3K\right\}.$ The power budgets per subcarrier are set to be $P^n=1,~n\in\N;$
   the desired transmission rate of all receivers are set to be $\gamma_k=B,~k\in\K;$
   and the channel gain and the noise power of all receivers on all subcarriers are set to be
   \begin{equation}\label{a}{\alpha_k^n}{}=2^{a_n}-1\geq0,~\eta_k^n=1,~k\in\K,\,n\in\N.\end{equation}
   Then the corresponding instance of the feasibility problem is \begin{equation}\label{construct}\left\{\!\!\!\!
\begin{array}{cl}
    & \displaystyle\sum_{n\in\N}\log_2\left(1+{\left(2^{a_n}-1\right)p_k^{n}}\right)\geq B,~k\in\K, \\
    & \eqref{ofdma}~\text{and}~1\geq p_k^n\geq 0,~k\in\K,~n\in\N.
    \end{array}\right.
\end{equation} We are going to show that the constructed problem \eqref{construct} is feasible if and only if the answer to the $3$-partition problem is yes, i.e., the set $\A$ can be partitioned into $K$ disjoint sets ${\cal S}_1,~{\cal S}_2,...,{\cal S}_K$ such that 
  \begin{equation}\label{partition}\sum_{n\in{\cal S}_k}a_n=B,~k=1,2,\ldots,K.\end{equation}


  We first show that if the $3$-partition problem has a yes answer, then the constructed problem \eqref{construct} is feasible.
   Suppose that $\A=\left\{a_1,a_2,...,a_{3K}\right\}$ can be partitioned into $K$ disjoint sets ${\cal S}_1,~{\cal S}_2,...,{\cal S}_K$ such that \eqref{partition} is true. 
   Then we can construct a feasible power allocation as follows: for each $k\in\K,$
  \begin{equation}\label{power}
  p_{k}^n=\left\{\begin{array}{lr}
       1, \quad \text{if}~n\in{\cal S}_k;\\[3pt]
       0, \quad \text{if}~n\notin{\cal S}_k.
    \end{array}\right.
  \end{equation}It is easy to see that the above power allocation vector \eqref{power} satisfies \eqref{ofdma} and $1\geq p_k^n\geq 0$ for all $k\in\K$ and $n\in\N$ in \eqref{construct}. Now, let us check the first condition of \eqref{construct}: for each $k\in\K,$
  \begin{equation*}\label{proof}\begin{array}{rl} &\displaystyle~~~\sum_{n\in\N}\log_2\left(1+{\left(2^{a_n}-1\right)p_k^{n}}\right)\\
                                                     &=\displaystyle\sum_{n\in{\S}_k}\log_2\left(1+{\left(2^{a_n}-1\right)}\right) ~\text{(from \eqref{power})}\\
                                                     &={\displaystyle\sum_{n\in{\S}_k}a_n=B}.~\text{(from \eqref{partition})}\\[15pt]
                                                     \end{array}\end{equation*}
  So, \eqref{power} is a feasible solution for problem \eqref{construct}.

  On the other hand, we show that if the constructed problem \eqref{construct} is feasible, then the answer to the $3$-partition problem is yes. Suppose that $\left\{\hat p_k^n\right\}$ is a power allocation vector which satisfies all conditions in problem \eqref{construct}. Clearly, we have
  \begin{equation}\label{eq1}\sum_{n\in\N}\log_2\left(1+\left(2^{a_n}-1\right)\hat p_k^n\right)\geq B,~k\in\K.\end{equation}
 Consider the sum-rate maximization problem 
   \begin{equation}\label{sum0}
 \begin{array}{cl}
\displaystyle \max_{\{p_k^n\}} & \displaystyle \frac{1}{K}\sum_{k\in\K}\sum_{n\in\N}\log_2\left(1+\left(2^{a_n}-1\right) p_k^n\right) \\
\text{s.t.} & \eqref{ofdma}~\text{and}~1\geq p_k^n\geq 0,k\in\K,n\in\N.
    \end{array}
\end{equation}Noticing that the channel gains from the transmitter to all receivers on each subcarrier are same in \eqref{sum0}, one can simply verify the following two useful facts:
 \begin{itemize}
   \item \textbf{Fact 1:} The optimal value of problem \eqref{sum0} is equal to $B;$\vspace{0.1cm}
   \item \textbf{Fact 2:} To achieve the optimal value $B$ of problem \eqref{sum0}, the total transmission power should be equal to $3K.$
 \end{itemize}
Combining \eqref{eq1} with \textbf{Fact 1}, we obtain
\begin{equation}\label{neweq1}\sum_{n\in\N}\log_2\left(1+\left(2^{a_n}-1\right)\hat p_k^n\right)=B,~k\in\K.\end{equation}
This, together with \textbf{Fact 2} and \eqref{ofdma}, implies \begin{equation}\label{neweq2}\hat p_k^n\in\left\{0,1\right\},~k\in\K,~n\in\N,\end{equation}and \begin{equation}\label{neweq3}\sum_{k\in\K}\sum_{n\in\N}\hat p_k^n=3K.\end{equation}
By defining
$${\S}_k=\left\{n\,|\,\hat p_k^n=1\right\},~k\in\K,$$ and using \eqref{neweq1}, \eqref{neweq2}, and \eqref{neweq3}, we immediately have \eqref{partition}, which shows the $3$-partition problem has a yes answer.
%

%
%
%
%

 Since the $3$-partition problem is strongly NP-complete\cite{Complexitybook}, we conclude that checking the feasibility of problem \eqref{problem} is NP-hard. Therefore, the optimization problem \eqref{problem} is NP-hard.\end{proof}

{Remark: The NP-hardness of problem \eqref{problem} has been shown in a recent paper \cite{add1}. In this letter, we provide a new proof, which is much simpler than the one in \cite{add1}. More important, the new proof can be directly extended to show the NP-hardness of problem \eqref{utility} with $H=H_2,H_3.$  
}

%

 \begin{dingli}\label{thm-hard2}
   Given any constant $c>1,$ the system utility maximization problem \eqref{utility} with $H=H_2,H_3,H_4$ are all NP-hard when $N/K=c$.
 \end{dingli}
 \begin{proof}
   {Without loss of generality,} consider the case $c=3.$ For any given instance of the $3$-partition problem, we construct the same downlink OFDMA system as in the proof of Theorem \ref{thm-power} and set $P=3K$ in problem \eqref{utility}.

   Theorem \ref{thm-power} directly implies that the min-rate maximization problem
   \begin{equation}
 \begin{array}{cl}
\displaystyle \max_{\{p_k^n\}} & \displaystyle H_4(R_1,R_2,...,R_K)\nonumber \\
\text{s.t.}  & \displaystyle \sum_{k\in\K}\sum_{n\in\N}p_k^{n}\leq 3K, \\
             & \eqref{ofdma}~\text{and}~1\geq p_k^n\geq 0,k\in\K,n\in\N.
    \end{array}
\end{equation}
is NP-hard, since the problem of checking whether its optimal value is greater than or equal to $B$ is NP-hard.

Consider the sum-rate maximization problem under the same setting, i.e.,
   \begin{equation}\label{sum}
 \begin{array}{cl}
\displaystyle \max_{\{p_k^n\}} & \displaystyle H_1(R_1,R_2,...,R_K) \\
\text{s.t.} & \displaystyle \sum_{k\in\K}\sum_{n\in\N}p_k^{n}\leq 3K, \\
             & \eqref{ofdma}~\text{and}~1\geq p_k^n\geq 0,k\in\K,n\in\N.
    \end{array}
\end{equation}We know from \textbf{Fact 1} and \textbf{Fact 2} in the proof of Theorem \ref{thm-power} that the optimal value of problem \eqref{sum} is equal to $B.$ 


Now, consider the utility maximization problem \eqref{utility} with $H=H_2$ and $H=H_3$ under the same setting, i.e.,
\begin{equation}\label{harmonic}
 \begin{array}{cl}
\displaystyle \max_{\{p_k^n\}} & \displaystyle H_2(R_1,R_2,...,R_K)\text{~or~}H_3(R_1,R_2,...,R_K) \\
\text{s.t.} 
    &  \displaystyle \sum_{k\in\K}\sum_{n\in\N}p_k^{n}\leq 3K, \\
             & \eqref{ofdma}~\text{and}~1\geq p_k^n\geq 0,k\in\K,n\in\N.
    \end{array}
\end{equation}
 Notice that for all $R_1$, $R_2$, \ldots, $R_K\geq 0,$
\begin{align*}&H_4(R_1,R_2,...,R_K)\leq H_3(R_1,R_2,...,R_K)\\
\leq ~& H_2(R_1,R_2,...,R_K)\leq H_1(R_1,R_2,...,R_K)\end{align*}
 and the equalities hold if and only if $R_1=R_2=\cdots=R_K.$ Therefore, the optimal value of problem \eqref{harmonic} is greater than or equal to $B$ if and only if the answer to the $3$-partition problem is yes. This implies the NP-hardness of problem {\eqref{utility} with $H=H_2$ and $H=H_3$}.
 \end{proof}

\subsection{Easy Cases}\label{sub2}
In this subsection, we identify some polynomial time solvable subclasses of problems \eqref{problem} and \eqref{utility}. 
\begin{dingli}\label{thm-easy}
1) The system utility maximization problem \eqref{utility} with $H=H_1$ is polynomial time solvable.

2) The system utility maximization problem \eqref{utility} is polynomial time solvable when there is only a single receiver.

3) The total power minimization problem \eqref{problem} is polynomial time solvable when either there is only a single receiver or the number of subcarriers is equal to the number of receivers.
  \end{dingli}
\begin{proof}
  1). We propose the following two-stage (subcarrier allocation stage and power allocation stage) polynomial time algorithm for solving the sum-rate maximization problem \begin{equation}\label{sumgeneral}
 \begin{array}{cl}
\displaystyle \max_{\{p_k^n\}} & \displaystyle \sum_{k\in\K} \sum_{n\in\N}{\log_2}\left(1+\frac{\alpha_{k}^n}{\eta_k^n}p_k^{n}\right) \\
\text{s.t.} & \displaystyle \eqref{subcarriercons}, \eqref{ofdma},~\text{and}~\sum_{k\in\K}\sum_{n\in\N}p_k^{n}\leq P.
    \end{array}
\end{equation}
\begin{center} \framebox{
\begin{minipage}{15.5cm}
\flushright
\begin{minipage}{15.5cm}
\centerline{\bf A Two-Stage Polynomial Time Algorithm for Problem \eqref{sumgeneral}}\vspace{0.1cm}
\begin{enumerate}
\item [\textbf{S1.}]  Subcarrier allocation: for each subcarrier $n\in\N,$
let
{{\begin{equation}\label{association}\!\!\!\pi(n)=\arg\max\left\{\frac{\alpha_{1}^n}{\eta_1^n},\frac{\alpha_{2}^n}{\eta_2^n},...,\frac{\alpha_{K}^n}{\eta_K^n}\right\}\footnote{Such $\pi(n)$ might not be unique, and if so, we choose any one of them.},~\frac{\alpha^n}{\eta^n}=\frac{\alpha_{\pi(n)}^n}{\eta_{\pi(n)}^n}.\end{equation}}}
\item
[\textbf{S2.}] Power allocation: solve problem
{{\begin{equation}\label{uK=1}
\begin{array}{rl}
\displaystyle \left\{(p^n)^*\right\}=\arg\max_{\{p^n\}} & \displaystyle \sum_{n\in\N}\log_2\left(1+\frac{\alpha^n}{\eta^n}p^{n}\right)\\[15pt]
\text{s.t.} & \displaystyle \sum_{n\in\N}p^n\leq P, \\
    &  P^n\geq p^n\geq 0,~n\in\N.
    \end{array}\end{equation}}}
\item [\textbf{S3.}] Output the optimal solution to problem \eqref{sumgeneral}:
{{\begin{equation}\label{solution}
{\left(p_{k}^n\right)^*}=\left\{
\begin{array}{cl}
(p^n)^*,&\mbox{if $k=\pi(n);$}\\[3pt]
0,&\mbox{otherwise,}
\end{array}\right.,~n\in\N.
\end{equation}}}
\end{enumerate}
\end{minipage}
\end{minipage}
}
\end{center}

The relation \eqref{association} indicates that the receiver with the best channel condition will be served on each subcarrier. Define $$\N_k=\left\{n\,|\,k=\pi(n)\right\},~k\in\K.$$ Then, the set of subcarriers $\N$ is optimally partitioned into $K$ nonoverlapping groups $\left\{\N_k\right\}_{k=1}^K$ (the transmitter will transmit signals to receiver $k$ on subcarriers in $\N_k$). The inequality \begin{equation}\label{algopt}
\begin{array}{rl}
    \displaystyle\sum_{n\in\N}\sum_{k\in\K} {\log_2}\left(1+\frac{\alpha_{k}^n}{\eta_k^n}p_k^{n}\right)\overset{(a)}{\leq} &\displaystyle\sum_{n\in\N}\sum_{k\in\K} {\log_2}\left(1+\frac{\alpha^n}{\eta^n}p_k^{n}\right)\\
    \overset{(b)}{\leq} & \displaystyle\sum_{n\in\N}\log_2\left(1+\frac{\alpha^n}{\eta^n}(p^{n})^*\right)\\
    \overset{(c)}{=} & \displaystyle\sum_{n\in\N}\sum_{k\in\K}{\log_2}\left(1+\frac{\alpha_{k}^n}{\eta_k^n}\left(p_k^{n}\right)^*\right)
    \end{array}
\end{equation} shows that $\left\{\left(p_k^n\right)^*\right\}$ returned by the proposed algorithm is globally optimal to problem \eqref{sumgeneral}, where $\left\{p_k^n\right\}$ is any feasible point of problem \eqref{sumgeneral}, $(a)$ is due to \eqref{association}, and $(b)$ is due to \eqref{uK=1} and the OFDMA constraint (at most one of $\left\{p_k^n\right\}_{k\in\K}$ is positive for any $n\in\N$), and $(c)$ is due to the construction of $\left(p_k^n\right)^*$ in \eqref{solution}.


Now, we show the polynomial time complexity of the proposed algorithm. It is obvious to see that the subcarrier allocation step \textbf{S1} takes $KN$ comparison operations. Moreover, we know from {\cite[Section IV-B]{multi-user}} that problem \eqref{uK=1} in the power allocation step \textbf{S2} can be solved in $N\log_2(N)$ operations by the extended water-filling algorithm. Therefore, {the worst-case complexity of the proposed algorithm is $O(NK+N\log_2(N)).$} 

 2). When there is only a single receiver ($K=1$) in the system, all the four system utility functions coincide and problem \eqref{utility} becomes problem \eqref{uK=1}. Therefore, problem \eqref{utility} with $K=1$ can be solved in $N\log_2(N)$ operations.

 3). Case 3) is a generalization of the {results} in \cite{multi-user} where polynomial time solvability of problem \eqref{problem} for the multi-user OFDMA system is studied. Case 3) can be proved by a similar argument as in \cite{multi-user}.
\end{proof}

{Table I summarizes the complexity status of the joint subcarrier and power allocation problems \eqref{problem} and \eqref{utility} for different scenarios.}

\begin{table*}[!ht]
{\caption{\textsc{Summary of the Complexity status of the Joint Subcarrier and Power Allocation Problem}}}
\label{table} \centering
\begin{tabular}[h]{|c|c|c|c|}
\hline
  \hline
        {\backslashbox{Scenario}{Problem}} &
{Total Power Min.} &  {Sum-Rate Max. ($H=H_1$)} &  {Utility Max. ($H=H_2,H_3,H_4$)}
\\\hline
        {Multi-User IC with Fixed $N>2$}  &   {NP-hard \cite{complexity}}   &    {NP-hard \cite{complexity}}      &     {NP-hard \cite{complexity}}
\\\hline
        {Multi-User OFDMA with $N/K>1$} &  {NP-hard\cite{multi-user}}   &        {NP-hard\cite{hayashi,multi-user}}    &    {NP-hard\cite{multi-user}}       \\\hline
        {Cellular Downlink OFDMA with $N/K>1$} &  {NP-hard\cite{add1}}   &        {Poly.\ Time Solvable (Theorem \ref{thm-easy})}    &    {NP-hard (Theorem \ref{thm-hard2})}\\\hline
        \hline
\end{tabular}
\end{table*}

\section{Concluding Remarks}\label{sec-conclusion}
In this letter, we have shown that the joint subcarrier and power allocation problem for the cellular downlink OFDMA system is generally NP-hard.
We have also identified some subclasses of the problem which are polynomial time solvable, such as the sum-rate maximization problem.
These complexity results reveal that the joint subcarrier and power allocation problem is intrinsically difficult to solve (except some special cases) and therefore provide valuable information to algorithm designers in directing their efforts toward those approaches that have the greatest potential of leading to useful algorithms.




\end{document}